\documentclass[aps,pra,amsmath,amssymb, reprint,superscriptaddress, floatfix]{revtex4-2}
\newcommand{\bol}[1]{\boldsymbol{#1}}
\usepackage{physics}
\usepackage{calc}
\usepackage{mathtools}
\usepackage{comment}
\newcommand{\cmmnt}[1]{}
\usepackage{tikz}

\usetikzlibrary{quantikz}
\usepackage{enumitem}
\usepackage{amsthm}
\theoremstyle{definition}
\newtheorem{theorem}{Theorem}
\newtheorem{definition}{Definition}
\newtheorem{lemma}[theorem]{Lemma}
\usepackage{balance}
\usepackage{longtable}
\usepackage[pdftex,pdfpagelabels,bookmarks,hyperindex,hyperfigures,colorlinks]{hyperref}

\begin{document}


\title{Symplectic perspective to quantum computing for Hamiltonian systems}


\author{Efstratios Koukoutsis}
\email{stkoukoutsis@mail.ntua.gr}
\affiliation{ColibriTD, 91 Rue du Faubourg Saint Honor{\'e}, 75008 Paris, France}
\affiliation{School of Electrical and Computer Engineering, National Technical University of Athens, Zographou 15780, Greece}
\author{Kyriakos Hizanidis}
\affiliation{School of Electrical and Computer Engineering, National Technical University of Athens, Zographou 15780, Greece}
\author{Lucas I I{\~n}igo Gamiz}
\author{{\'O}scar Amaro
}
\affiliation{GoLP/Instituto de Plasmas e Fus{\~a}o Nuclear, Instituto Superior T{\'e}cnico, Universidade de Lisboa,
1049-001 Lisbon, Portugal}
\author{Christos Tsironis}
\affiliation{School of Electrical and Computer Engineering, National Technical University of Athens, Zographou 15780, Greece}
\author{Abhay K. Ram}
\affiliation{Plasma Science and Fusion Center, Massachusetts Institute of Technology, Cambridge,
Massachusetts 02139, USA}
\author{George Vahala}
\affiliation{Department of Physics, William \& Mary, Williamsburg, Virginia 23187, USA}


\date{\today}

\begin{abstract}
This work develops a symplectic framework for quantum computing to be applied to classical Hamiltonian systems, exploiting the intrinsic geometric compatibility between unitary quantum evolution and symplectic phase-space dynamics in a two-fold way. The first part is devoted in establishing an exact correspondence between quantum evolution and classical Hamiltonian flow on a Kähler manifold. This correspondence enables a geometric quantization scheme that identifies a family of classical Hamiltonian systems admitting exponentially compressed quantum representations--appropriate for quantum simulation. In the second part we demonstrate that Liouville-integrable Hamiltonian dynamics induce finite-dimensional unitary evolution through action–angle variables and Koopman–von Neumann encoding. This allows efficient quantum representation and parallel evolution of large phase-space ensembles, where entangled encodings provide exponential compression in ensemble size and enable quantum speed-ups in observable estimation via amplitude estimation techniques.
For non-integrable systems,  Lie canonical perturbation theory is incorporated to construct near-symplectic transformations that map dynamics to approximately integrable forms, preserving unitary evolution up to a controlled error. We derive the resulting quantum computational complexity of the proposed quantum-symplectic scheme, revealing both an exponential compression in memory requirements and a potential polynomial speed-up with respect to the system size. Finally, the transport evolution equation governing the quantum phase-space observables is obtained.
\end{abstract}


\maketitle

\section{Introduction}\label{sec:Introduction}

Quantum computers operate under the laws of closed quantum systems, and therefore inherit a linear structure that only allows for unitary operations--gates. General linear transformations can nevertheless be implemented through block encoding techniques~\cite{Childs_2012,Schlimgen_2021,Schlimgen_2022,Jin_2023,Koukoutsis_2025}, where the non-unitary operator is embedded into a larger unitary one.

From a group theoretic point of view, the $N$-dimensional unitary operators $U(N)$ lie at the intersection of the generalized linear group $GL(\mathbb{C}, N)$, the orthogonal group $O(2N)$ and the symplectic group $Sp(\mathbb{R}, 2N)$,
\begin{equation}\label{group unitary}
U(N)=O(2N)\cap GL(\mathbb{C}, N)\cap Sp(\mathbb{R}, 2N).
\end{equation}
Whereas the complex linear and orthogonal structure of unitary quantum computing has motivated the development of quantum algorithms for simulating various linear processes in physics~\cite{Harrow_2009,Berry_2014,Koukoutsis_2023,An_2023,Babbush_2023,Novikau_2023,Bosch_2025,Amaro_2025,Wawrzyniak_2025,Bhuvanesh_2026,Gamiz_2026}, the symplectic part of the threefold structure~\eqref{group unitary} and its implications on  quantum computing for symplectic systems remain relatively unexplored. In this paper we examine the symplectic nature of unitary operations for a highly important  class of symplectic systems in physics, namely finite-dimensional Hamiltonian systems. 

Hamiltonian systems are a broad class of dynamical systems whose evolution take place in the $2N$-dimensional phase-space of conjugate generalized positions and momenta $(\bol q, \bol p)$. Their dynamics are generated by a scalar phase-space function $H(\bol q, \bol p, t)$, called Hamiltonian, that usually represents the total energy of the system. The time evolution is governed by Hamilton's equations in canonical form:
\begin{equation}\label{canonical Hamilton equations}
\dv{q_k}{t}=\pdv{H}{p_k},\quad \dv{p_k}{t}=-\pdv{H}{q_k},\quad k=0,1,\ldots,N-1.
\end{equation}
Hamiltonian systems lie at the heart of classical mechanics, describing a plethora of physical systems from simple oscillators to celestial dynamics~\cite{Arnold_1989, Lichtenberg_1992,Saletan_1998}. An important feature of Hamiltonian dynamics is that they inherently possess a symplectic structure $\omega$, as a two-form~\cite{Arnold_1989}
\begin{equation}\label{symplectic two-form}
\omega=\sum_{k=0}^{N-1} dq_k\wedge dp_k,
\end{equation}
that is preserved by the Hamiltonian flow in Eq.~\eqref{canonical Hamilton equations} and geometrically encodes the canonical relationships between positions and momenta
\begin{equation}
\{q_k,q_m\}=\{p_k,p_m\}=0,\quad \{q_k, p_m\}=\delta_{km},
\end{equation}
for $ k,m=0,1,\ldots, N-1$. The $\{\cdot,\cdot\}$ bilinear operation denotes the Poisson bracket defined as,
\begin{equation}\label{poisson braket}
\{f,g\}=\sum_k \pdv{f}{q_k}\pdv{g}{p_k}-\pdv{f}{p_k}\pdv{g}{q_k}.
\end{equation}
In that respect, a phase-space transformation $M:\bol z\to\Bar{\bol z}$ in the phase-space variables $\bol z=(\bol q, \bol p)$ is symplectic iff its Jacobian matrix $J_M=\pdv{\Bar{\bol z}}{\bol z}$ preserves the canonical symplectic structure,
\begin{equation}\label{definition of symplectic}
J_M^T\Omega J_M=\Omega\Leftrightarrow M\in Sp(\mathbb{R},2N)
\end{equation}
where $\Omega$ is the matrix representation of the $\omega$ form,
\begin{equation}
\Omega =
\begin{pmatrix}
0_{N\times N} & I_{N\times N}\\
-I_{N\times N} & 0_{N\times N}
\end{pmatrix}.
\end{equation}

In general, Hamilton's equations are nonlinear in the phase-space variables $\bol z=(\bol q, \bol p)$. As a result, extracting the nonlinear Hamiltonian flow $\bol z(t)=\Phi_t[\bol z_0]$ from Eqs.~\eqref{canonical Hamilton equations}, where $\Phi_t$ is the nonlinear evolution operator, using a linear quantum algorithm is conceptually challenging.  To alleviate this difficulty,  a common strategy is to employ infinite linear embedding of the nonlinear dynamics~\cite{Joseph_2020,Wu_2025} with an appropriate truncation to recover a finite $2N$-dimensional representation while ensuring the symplectic structure preservation in Eq.~\eqref{definition of symplectic}.

In this work, we explore an alternative pathway for leveraging quantum computing to study nonlinear Hamiltonian systems. Our approach exploits the common structural link between the symplectic component of unitary quantum operations and the intrinsic symplectic geometry of Hamiltonian dynamics. Building on this connection, we provide a symplectic quantum framework that captures long-time, high-fidelity aspects of complex nonlinear Hamiltonian systems through two complementary perspectives.

Firstly, a geometric formulation of quantum mechanics in a Kähler manifold, that compactly encompasses a positive definite Fubini–Study metric and a symplectic form, is  presented~\cite{Heslot_1985,Ashtekar_1999,Brody_2001,Volovich_2025}. Within this real and geometric framework the unitary flow generated by the quantum Schrodinger equation is recast as specific classical canonical Hamiltonian flow. Conversely, by leveraging this mapping, we identify an entire family of classical quadratic Hamiltonian systems that admit a geometric quantization, thereby transforming them into a quantum system that can potentially be solved faster on a quantum computer.

Secondly, we demonstrate that the constraint of Liouville integrability induces a unitary flow in symplectic phase-space. Building on this property, and exploiting quantum parallelism together with Liouville’s theorem, we propose a quantum framework for time evolution of microscopic phase-space ensembles for probing collective Hamiltonian dynamics as quantum observables via quantum amplitude estimation. Extending our approach to non-integrable systems, we review the Lie canonical perturbation theory as a method for symplectic integration via the construction of a hierarchy of integrable approximations~\cite{Kominis_2008}. This enables the application of the previous quantum workflow for capturing the coherent dynamics of nonlinear and non-integrable Hamiltonian systems in phase-space. The complexity of the proposed quantum unitary and symplectic scheme is derived and compared with the classical symplectic integrators, revealing potential advantages in computational resources and runtime as a function of the ensemble total degrees of freedom in the phase space.

The structure of this paper is as follows. Sections~\ref{sec:2.1},~\ref{sec:2.2} establish the Kähler geometric representation for the quantum states and operators. Then, in Sec.~\ref{sec:2.3} we show that quadratic Hamiltonians with Kähler structure admit an exact quantum representation. Sections~\ref{sec:3.1} and~\ref{sec:3.2} establish the unitary dynamics of integrable systems in the phase-space under the action-angle coordinates and detail the quantum encoding for leveraging quantum parallelism. Section.~\ref{sec:3.3}, connects quantum sampling of discrete phase-space trajectories with averaged classical observables and collective dynamics. In Sec.~\ref{sec:4.1}, the required elements of Lie canonical perturbation theory are outlined. The validity and the computational complexity of the Lie unitary approximation are examined in Sec.~\ref{sec:4.2}. Finally, a form of the evolution equation for the quantum observables is analytically derived in Sec.~\ref{sec:4.3}.

\section{Symplectic structure of Schrodinger dynamics}
The starting point of the discussion that will follow within this section is the Schrodinger equation which describes the dynamics of closed, pure-state quantum systems:
\begin{equation}\label{Schrodinger Eq}
i\pdv{\ket{\psi}}{t}=\hat{H}\ket{\psi},\quad \ket{\psi}\in\mathcal{H}(\mathbb{C},N),\quad \hat{H}=\hat{H}^\dagger\in\mathcal{D}(\mathcal{H}).
\end{equation}
In Eq.~\eqref{Schrodinger Eq}, the pure (normalized) quantum states $\ket{\psi}$ live in the $N$-dimensional Hilbert space $\mathcal{H}$, and the Hamiltonian operator $\hat{H}$ is self-adjoint into a dense  domain $\mathcal{D}(\mathcal{H})$ to generate unitary evolution. Given an orthonormal basis set $\{\ket{k}\}\in\mathcal{H}$, every state vector can be written as,
\begin{equation}\label{quantum state}
\ket{\psi}=\sum_{k=0}^{N-1}\psi_k\ket{k},\quad \psi_k\in\mathbb{C}.
\end{equation}
Finally, the Hilbert space $\mathcal{H}$ is equipped with a bilinear inner product structure $\braket{\cdot}{\cdot}:\mathcal{H}\times\mathcal{H}\to\mathbb{C}$,
\begin{equation}\label{inner prod}
\braket{\phi}{\psi}=\sum_{k,l=0}^{N-1}\phi_l^*\psi_k.
\end{equation}

\subsection{The Strocchi map: from Hilbert to Kähler space}\label{sec:2.1}
\begin{definition}[Kähler space]
Let $K$ be real vector space endowed with the following structures:
\begin{align}
&g:K\times K\to \mathbb{R}_+,\,\, \text{(positive definite bilinear form)},\label{kahler bilinear form}\\
&\omega:K\times K\to \mathbb{R},\,\, \text{(symplectic form)},\label{khaler symplectic form}\\
& \mathcal{J}:K \to K,\,\,\mathcal{J}^2=-id,\,\, \text{(complex structure)}\label{Khaler imaginary unit},
\end{align}
where $id$ is the identity element of $K$. Then, the quadruplet $(K,g,\omega,\mathcal{J})$ defines a  $2N$-dimensional real Kähler space $\mathcal{K}$ with $N=dim K$.
\end{definition}

Therefore, the Kähler space $\mathcal{K}$ defines a manifold with three compatible structures: Riemannian, symplectic and complex. Also the identities $g(\cdot,\cdot)=\omega(\cdot,\mathcal{J}\cdot)$ and $\omega(\mathcal{J}\cdot, \mathcal{J}\cdot)=\omega(\cdot,\cdot)$ which relate the different structures,
hold.

\begin{definition}[Strocchi map~\cite{Strocchi}]
The Strocchi map is an invertible realification map from the complex Hilbert space $\mathcal{H}$ to the real Kähler space $\mathcal{K}$,
\begin{equation}\label{abstract Strocci map}
\mathcal{S}:\mathcal{H}\to\mathcal{K}.
\end{equation}
\end{definition}
Under the representation  $\psi_k=\psi_k^R+i\psi_k^I$ of the quantum amplitudes in Eq.~\eqref{quantum state},
\begin{equation}\label{Khaler space elements}
\mathcal S(\psi_k)=(\psi_k^R, \psi_k^I)\equiv(q_k,p_k)=\bol{z}_k\in\mathbb{R}^{N}\times\mathbb{R}^{N},
\end{equation}
and the inner product structure in Eq.\eqref{inner prod} is expressed in terms of the Kähler structures as
\begin{equation}\label{Khaler inner product}
\braket{\phi}{\psi}=g(\bol{z}_\phi,\bol z_\psi)+i\omega(\bol{z}_\phi,\bol z_\psi).
\end{equation}
The real part $g$ of the inner product in Eq.~\eqref{Khaler inner product} is directly related (up to a factor that ensures consistency with quantum transition probabilities between states and with the geodesic distance on projective space) to the Fubini-Study (F-S) metric~\cite{Brody_2001},
\begin{equation}\label{f-s metric}
g(\bol{z}_\phi,\bol{z}_\psi)
= \langle \bol q_\phi|\bol q_\psi\rangle_{\mathbb R}
+ \langle \bol p_\phi|\bol p_\psi\rangle_{\mathbb R}
\end{equation}
where, the following notation has been employed:
\begin{equation}
\begin{aligned}
&\langle {\bol a}|{\bol b}\rangle_{\mathbb R}=\sum_{K=0}^{N-1}a_kb_k,\quad{\bol a},{\bol b}\in \mathbb R^N.
\end{aligned}
\end{equation}
The F-S metric is always positive definite. On the other hand, the imaginary part $\omega$ in Eq.~\eqref{Khaler inner product} is the symplectic form,
\begin{equation}\label{symplectic}
\omega(\bol{z}_\phi,\bol{z}_\psi)
= \langle \bol q_\phi|\bol p_\psi\rangle_{\mathbb R}
- \langle \bol p_\phi|\bol q_\psi\rangle_{\mathbb R}.
\end{equation}
In terms of the selected map $\mathcal{S}$, the induced structures $(g,\omega,\mathcal{J})$ also read:
\begin{align}
&g(\bol{z}_\phi,\bol z_\psi)=\omega(\bol{z}_\phi,\mathcal{J}\bol z_\psi),\\
&\omega(\bol{z}_\phi,\bol z_\psi)=\bol{z}_\phi^T\Omega\bol z_\psi,\\
&\mathcal{J}\bol z=(-\bol p, \bol q),
\end{align}
with block-matrix representations:
\begin{equation}\label{J,OMEGA representation}
\mathcal{J}=
\begin{pmatrix}
0_{N\times N} & -I_{N\times N}\\
I_{N\times N} & 0_{N\times N}
\end{pmatrix},\quad 
\Omega =
\begin{pmatrix}
0_{N\times N} & I_{N\times N}\\
-I_{N\times N} & 0_{N\times N}
\end{pmatrix}.
\end{equation}

Subsequently, under the Strocchi map $\mathcal{S}$, the quantum states $\ket{\psi}$ in the inner-product $N$-dimensional complex Hilbert space $\mathcal{H}$ are translated into the $\bol z$ elements belonging to the real, $2N$-dimensional Kähler space as presented in Fig.~\ref{fig:1}.
\begin{figure}[h]
    \centering
    \includegraphics[width=\linewidth]{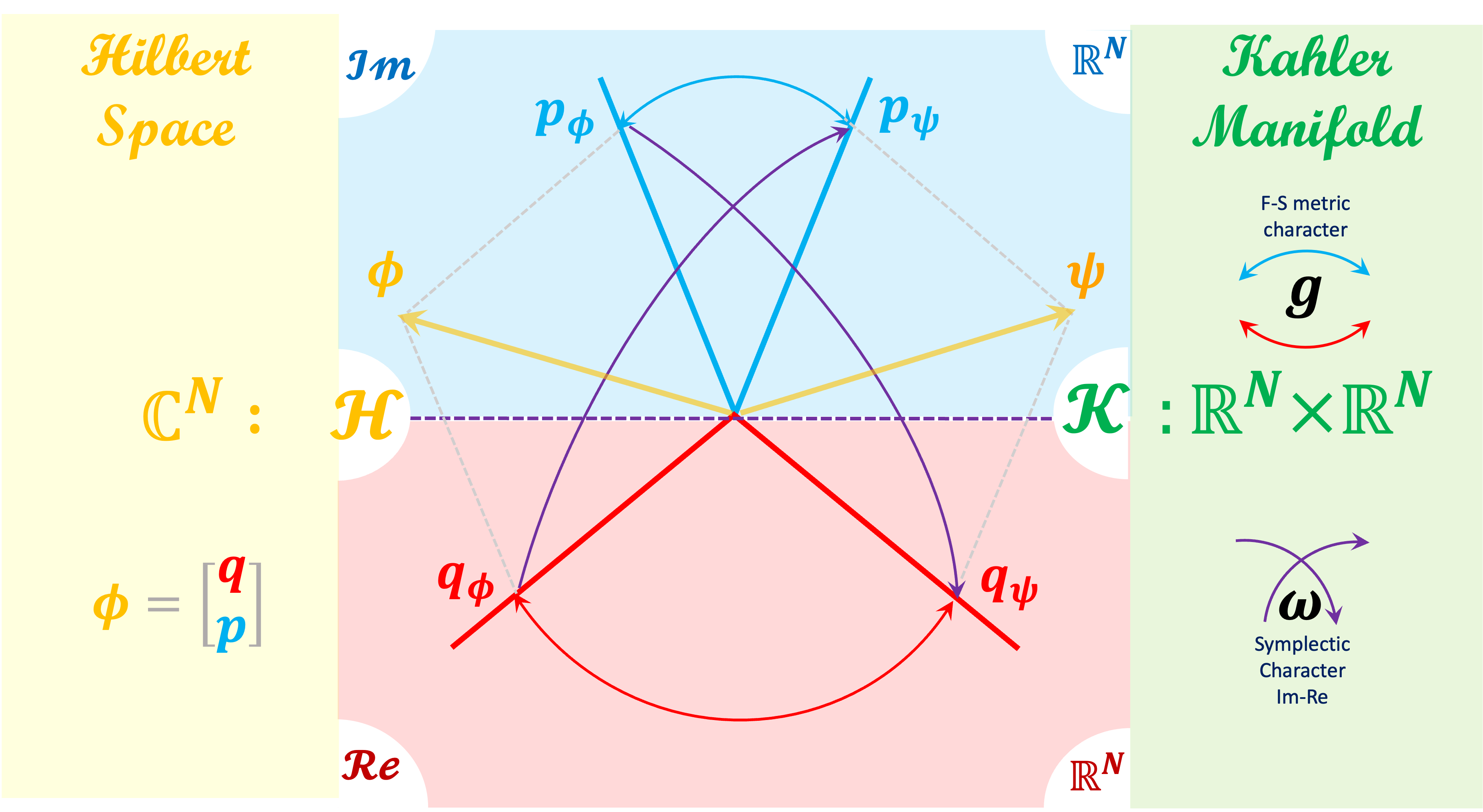}
    \caption{Schematic of the Hilbert space representation as a Kähler
manifold via the Strocchi map. The F-S metric $g=\overset{\Large\frown}{q_\phi q_\psi}+\overset{\Large\frown}{p_\phi p_\psi}$ describes the geodesic distances, hence the arcs on top of the pairs (red and blue), among the respective components in the $2N$-dimensional  real space. The order does not matter since the components involved are of the same kind. The symplectic structure $\omega=\overset{\Large\curvearrowright}{q_\phi p_\psi}-\overset{\Large\curvearrowright}{p_\phi q_\psi}=\overset{\Large\curvearrowright}{q_\phi p_\psi}+\overset{\Large\curvearrowright}{q_\psi p_\phi}$ indicates the ``cross" geodesic distances (in purple) between elements of the opposite kind, and therefore the order does matter. The state vectors $\phi$ and $\psi$ (orange) live in the Hilbert space, while $q$'s (in red) and $p$'s (in blue) in separate $N$-dimensional real spaces indicated with respective colors.}
    \label{fig:1}
\end{figure}
In Equation~\eqref{Khaler space elements},  the real elements $\bol{z}\in\mathcal{K}$ are expressed in terms of the canonical positions  $\bol q$ and momenta $\bol p$, as we will associate the quantum Schrodinger flow with a corresponding classical Hamiltonian flow in the $(\bol q,\bol p)$-phase-space. 

\subsection{Operators and the induced classical Hamiltonian}\label{sec:2.2}
Beyond describing quantum states in $\mathcal{K}$ it is important to encode quantum operators within the same framework. Following the Strocchi map in Eq.~\eqref{Khaler space elements},  any linear quantum operator $\hat{A}\in GL(\mathbb{C}, N)$ with  matrix representation $A\in Mat(\mathbb{C},N)$, can be decomposed into real and imaginary parts,
\begin{equation}\label{matrix A decompsition}
A=A_R+iA_I,\quad A_R,A_I\in Mat(\mathbb{R},N).
\end{equation}
Therefore, we define the operator mapping $\Tilde{A}:\mathcal{K}\to\mathcal{K}$ in the Kähler representation as,
\begin{equation}\label{Kahler operator matrix representation}
\Tilde{A}=\begin{bmatrix}
A_R&-A_I\\
A_I& A_R
\end{bmatrix}.
\end{equation}
Notice that The Kähler matrix $\Tilde{A}\in Mat(\mathbb{R}, 2N)$ is a real matrix compatible with the complex structure, i.e., it satisfies
\begin{equation}\label{complex structure commutation}
\Tilde{A}\mathcal{J}=\mathcal{J}\Tilde{A},
\end{equation}
where the representation of $\mathcal{J}$ is given in Eq.~\eqref{J,OMEGA representation}.

 In standard quantum mechanics self-adjoint operators $\hat{O}=\hat{O}^\dagger$ correspond to observable physical quantities, such that $\mel{\psi}{\hat{O}}{\psi}\in\mathbb{R}$. In the  Kähler representation, the corresponding operator $\Tilde{O}$  must satisfy, in addition to Eq.~\eqref{complex structure commutation},
\begin{equation}\label{Kahler hermitian}
\Tilde{O}=\Tilde{O}^T\Leftrightarrow O_R^T=O_R\,\,\text{and}\,\,O_I^T=-O_I.
\end{equation}
The second important class are the unitary operators $\hat{U}^{-1}=\hat{U}^\dagger$, that correspond to probability preserving transformations, constituting the admissible quantum computing gates. Their Kahler counterparts $\Tilde{U}$ satisfy,
\begin{equation}\label{Kahler unitary}
\Tilde{U}^{-1}=\Tilde{U}^T \,\, \text{and}\,\,\Tilde{U}^T\Omega\Tilde{U}=\Omega.
\end{equation}
Together with the commutation condition in Eq.~\eqref{complex structure commutation}, the relations in Eq.~\eqref{Kahler unitary} reiterate that unitary operators are represented in $\mathcal{K}$ as linear transformations that are simultaneously orthogonal, symplectic and compatible with the complex structure.

Out of the various observable operators, the Hamiltonian operator $\hat{H}$ has a central importance in quantum mechanics as it represents the energy observable of the system and generates unitary evolution via the Schrodinger equation~\eqref{Schrodinger Eq}.  Under the Strocchi map, the Schrodinger equation~\eqref{Schrodinger Eq} is mapped in the Kähler space $\mathcal{K}$ into the classical canonical Hamilton's equations~\cite{Strocchi}:
\begin{equation}\label{Hamilton-Schrodinger}
\dv{\bol z}{t}=\Omega\nabla_{\bol z}H_c\Leftrightarrow \dv{q_k}{t}=\pdv{H_c}{p_k},\quad \dv{p_k}{t}=-\pdv{H_c}{q_k},
\end{equation}
where $H_c$ is the induced classical Hamiltonian function (up to a prefactor of $1/2$),
\begin{widetext}
\begin{equation}\label{classical Hamiltonian}
 H_c(\bol q, \bol p)={1 \over 2}\sum_{k=0}^{N-1}\sum_{m=0}^{N-1}\Big[(q_kq_m+p_kp_m)\Re{H_{km}}+(p_kq_m-q_kp_m)\Im{H_{km}}\Big],\quad H_{km}=\mel{k}{\hat{H}}{m}.
\end{equation}
\end{widetext}
Similarly, any quantum observable $\mel{\psi}{\hat{O}}{\psi}$ is translated, through the Strocchi map, to a quadratic phase-space function $f(\bol q,\bol p)$. 

The induced classical quadratic Hamiltonian function $H_c$ in Eq.~\eqref{classical Hamiltonian} describes  $N$ coupled harmonic oscillators, subject to an effective interaction generated by the $p_kq_m-q_kp_m$ coupling terms. As a result, the unitary Schrodinger flow, generated by the Hamiltonian operator $\hat{H}$, in the complex  Hilbert space $\mathcal{H}$ of an $N$-dimensional quantum system can be equivalently mapped onto the classical Hamiltonian flow of Eq.~\eqref{Hamilton-Schrodinger} in a $2N$-dimensional phase-space with a classical quadratic Hamiltonian $H_c$ of $N$ coupled harmonic oscillators. 

\subsection{Kähler quantization and quantum computing implications}\label{sec:2.3}
Recasting the Schrodinger evolution \eqref{Schrodinger Eq} into classical canonical Hamilton equations \eqref{Hamilton-Schrodinger} provides a potential pathway to leverage this similarity in several ways. In particular, it enables  the application of symplectic integration techniques~\cite{Feng_2010} to $H_c$ for simulation of the Schrodinger equation, the  emulation of quantum gates and unitary dynamics using classical oscillators~\cite{Briggs_2012,Briggs_2013}, the formulation of a canonical perturbation theory for quantum systems~\cite{Manjarres}, and the application of classical Hamiltonian control methods on $H_c$ for quantum control purposes~\cite{Mendes_2003}.

However, there is a caveat related to dimensionality. According to the Strocchi map, an $n$-qubit quantum system with $n>>1$ possesses $N=2^n$ degrees of freedom and is mapped into an exponentially large system of $2^n$ coupled harmonic oscillators in a $2N = 2^{n+1}$-dimensional phase-space. Classically simulating the dynamics of $H_c$ for $N=2^n$ oscillators is as difficult as diagonalizing the quantum Hamiltonian operator $\hat{H}$, which scales as $\mathcal{O}(N^3) = \mathcal{O}(2^{3n})$ for $n>>1$. This can be seen by directly diagonalizing the $\hat{H}$ through the unitary operator $\hat{V}$,
\begin{equation}
\hat{H}=\hat{V}^\dagger\hat{\Delta}\hat{V},\,\,\hat{\Delta}=\lambda_{k}\delta_{km},\quad\lambda_k\in\mathbb{R},
\end{equation}
which in principle uncouples the oscillators and yields the normal form of the classical Hamiltonian $H_c$,
\begin{equation}\label{normal form}
\bar{H}_c=\frac{1}{2}\sum_{k=0}^{N-1}\lambda_k(q^2_k+p^2_k).
\end{equation}
Integration of Hamilton's equation for the $\bar{H}_c$ is now trivial since it represents separate classical harmonic oscillators with normal frequencies $\omega_k=\lambda_k$. In addition, from the normal form in Eq.\eqref{normal form} it is demonstrated that all quantum systems are integrable in a Liouville sense~\cite{Volovich_2019} since $\{I_k,I_m\}=0$ where $I_k=(q_k^2+p_k^2)/2$ for $k,m=0,1,\ldots,N-1$.

On the other hand, the established correspondence between the Schrodinger $\hat{H}$ Hamiltonian operator and induced classical Hamiltonian $H_c$ can be interpreted in reverse. Specifically, the following result identifies 
which classical stable Hamiltonian systems, described by a quadratic Hamiltonian $H$ in a $2N$-dimensional phase-space, can be mapped into the form of $H_c$ and subsequently rewritten as a $n=\log_2N$-qubit quantum Hamiltonian ${H}_q$.
\begin{lemma}[Kähler quantization]
Let $H=\bol z^T\Tilde{H}\bol z,$ be a stable classical quadratic Hamiltonian on a $2N$-dimensional phase-space, with $\Tilde{H}\in Mat(\mathbb{R}, 2N)$ a block-matrix of the form,
\begin{equation}\label{tilde H}
\Tilde{H}=\begin{pmatrix}
Q_1 & P^T \\
P &Q_2
\end{pmatrix},\,\,Q_{1,2},P\in Mat(\mathbb{R},N),\,\, Q_{1,2} = Q_{1,2}^T.
\end{equation}
If the classical Hamiltonian  matrix $\Tilde{H}$ possesses complex structure, i.e satisfying $\Tilde{H}\mathcal{J}=\mathcal{J}\Tilde{H}$, then the classical Hamiltonian flow generated by $H$ can be equivalently represented as a quantum Schrodinger evolution with the corresponding $N$-dimensional quantum Hamiltonian $H_q$ matrix being
\begin{equation}\label{quantal hamiltonian operator}
H_q=Q_1+iP\in Mat(\mathbb C,N).
\end{equation}
\end{lemma}
\begin{proof}
Explicitly calculating the complex symmetry condition $\Tilde{H}\mathcal{J}=\mathcal{J}\Tilde{H}$ using $\Tilde{H}$ in Eq.~\eqref{tilde H} and the matrix representation of $\mathcal{J}$ in Eq.~\eqref{J,OMEGA representation} we obtain $Q_1=Q_2=Q$ and $P^T=-P$. Therefore, according to Eq.~\eqref{Kahler operator matrix representation}, $\Tilde{H}$ has a Kähler structure,
\begin{equation}\label{Kahler classical structure}
\Tilde{H}=\begin{pmatrix}
Q & -P \\
P &Q
\end{pmatrix}.
\end{equation}
Applying the inverse Strocchi map $\mathcal{S}^{-1}:\mathcal{K}\to\mathcal{H}$, the classical Kähler Hamiltonian $\tilde{H}$ naturally induces a complex quantum Hamiltonian matrix  $H_q$,
\begin{equation}\label{quantal hamiltonian operator2}
H_q=Q+iP\in Mat(\mathbb C,N).
\end{equation}
\end{proof}
This establishes a one-to-one correspondence between the classical quadratic Hamiltonian system $H$ with a Kähler structure and $N$-dimensional quantum Hamiltonian matrix $H_q$.

Thus, the problem of simulating the $2N$-dimensional classical Hamiltonian flow of $H$  using standard symplectic techniques, is reduced to simulating the quantum Schrodinger equation
\begin{equation}\label{quantal Schrodinger}
i\pdv{\bol\psi}{t}=H_q\bol\psi, \quad \bol\psi=\bol q+i\bol p,
\end{equation}
corresponding to a $n=\log_2 N$ qubit system. By performing this quantum representation, one can leverage potentially exponentially faster quantum simulation techniques for Eq.~\eqref{quantal Schrodinger} which scale as $\mathcal{O}(\log_2N, T, \log(1/\epsilon))$ and allow efficient extraction of relevant observables~\cite{Babbush_2023,Dsouza_2006}. In contrast, direct classical symplectic methods scale as $\mathcal{O}(N, T, 1/\epsilon^\kappa)$, where $\kappa$ is the order of the symplecic integrator~\cite{Feng_2010}. Importantly, this geometric quantization is exact: any unitary quantum algorithm automatically preserves the underlying symplectic structure of Hamilton’s equations for $H$, ensuring full fidelity to the classical dynamics while achieving exponential compression.

\section{Liouville integrability and quantum computing}\label{sec:3}
In the previous section we established the equivalence between Schrodinger dynamics and the Hamiltonian flow generated by a classical integrable  and finite-dimensional quadratic Hamiltonian, and vice versa. In this section, we relax the quadratic constraint to investigate more general and nonlinear classical and finite-dimensional Hamiltonian systems, by examining how the Liouville integrability relates to quantum unitary evolution and exploring the resulting consequences for quantum computation.

\subsection{Liouville integrability generates unitary flow}\label{sec:3.1}
\begin{definition}[Action-angle variables~\cite{Arnold_1989}]
Consider a Liouville integrable Hamiltonian system $H(q_k,p_k)$ with $k=0,1,\ldots,N-1$ with compact and dense energy levels. Then, there exists a canonical transformation $T:(q_k,p_k)\to(\theta_k,I_k)$ to a set of action-angle variables $(I_k, \theta_k)$, generated by a function $W(q_k, I_k)$, such that the transformed Hamiltonian  $K$ depends only on the actions,
\begin{equation}\label{general integrability transformation}
H(q_k,p_k)\xrightarrow{T}K(I_k).
\end{equation}
The actions $I_k$ are provided by
\begin{equation}\label{action integral}
I_k=\frac{1}{2\pi}\oint_{\gamma_k} p_k d\,q_k,
\end{equation}
and the canonical transformation is defined through the generating function $W$ as,
\begin{equation}\label{transformation coordinates}
p_k=\pdv{W}{q_k},\quad \theta_k=\pdv{W}{I_k}.
\end{equation}
\end{definition}

In the action-angle representation, the flow generated by the $K$ Hamiltonian reads,
\begin{align}
&I_k(t)=I_k(t_0)>0, \label{action-angle evolution}\\
&\theta_k(t)=\theta_k(t_0)+\omega_k(t-t_0),\quad \omega_k=\pdv{K}{I_k} \label{action-angle evolution2}.
\end{align}
By encoding the action-angle variables $(\boldsymbol{\theta}, \boldsymbol{I})$ into a quantum state $\ket{\psi}$ using the Koopman–von Neumann (KvN) formalism~\cite{Joseph_2020,Manjarres},
\begin{equation}\label{KvN state}
\ket{\psi(t)}=\sum_{k=0}^{N-1}\sqrt{I_k(t)}e^{-i\theta_k(t)}\ket{k},
\end{equation}
the classical evolution equations~\eqref{action-angle evolution},~\eqref{action-angle evolution2} are compactly expressed as a unitary flow generated by a unitary diagonal operator
\begin{equation}\label{unitary diagonal integrable}
\hat{U}(t;t_0)=\sum_ke^{-i\omega_k(t-t_0)}\ket{k}\bra{k},\quad \ket{\psi(t)}=\hat{U}(t;t_0)\ket{\psi(t_0)}.
\end{equation}
For the quantum state in Eq.~\eqref{KvN state} to be a normalized, the actions must satisfy $\bol I\in S^{N-1}$.

Consequently, any classical integrable Hamiltonian system with an explicit action-angle representation, admits a unitary quantum evolution, where the eigenvalues of the corresponding quantum Hamiltonian operator $\hat{H}$ are precisely the canonical frequencies $\omega_k$.

As before, the quantum encoding of the action-angle variables of the integrable system in Eq.~\eqref{KvN state} achieves an exponential compression requiring $n=\log_2N$ qubits to encode $2N$ degrees of freedom. Implementing the diagonal unitary evolution in Eq.~\eqref{unitary diagonal integrable} exactly on a quantum computer requires in general $\mathcal{O}(2^n)=\mathcal{O}(N)$ elementary unitary gates~\cite{Bullock_2004}.

Two key elements emerge in the present discussion and must be highlighted. Firstly, the starting point is a general Liouville integrable system which can be nonlinear, such as the pendulum whose (dimensionless) Hamiltonian is 
\begin{equation}\label{Pendulum Hamiltonian}
H=\frac{p^2}{2}+1-\cos{\phi},
\end{equation}
and admits a local action-angle representation for both rotation and libration regimes~\cite{Brizzard_2013}. Despite the nonlinear nature of the original Hamiltonian and its associated flow, the canonical transformation to action-angle variables $T$ mediates the dynamics to a linear unitary evolution Eq.~\eqref{unitary diagonal integrable}. Secondly, integrability imposes a strong structural constraint: the KvN quantum state representation in Eq.~\eqref{KvN state} remains finite--$N$--dimensional even though the underlying system is nonlinear. This is because the dynamics are confined to  an invariant tori characterized by a discrete set of $N$ normal frequencies ${\omega_k}$, rather than requiring an arbitrary infinite-dimensional spectral decomposition.

These two factors establish that finite-dimensional nonlinear yet integrable and stable classical dynamics can be exactly recast into finite-dimensional unitary quantum evolution. Given the Liouville integrability of all quantum systems~\cite{Volovich_2019} (see also Sec.~\ref{sec:2.3}) this interconnection is somewhat expected.

\subsection{Probing phase-space with quantum parallelism}\label{sec:3.2}
Consider a discrete finite-dimensional ensemble of $N_s$ phase-space trajectories $\bol \rho(t)=\{I_k^j(t), \theta_k^j(t)\}$ within a phase-space volume $\mathcal{V}$ with $k=0,1,\ldots,N-1$ and $j=0,1,\ldots,N_s-1$. Then, $\bol\rho$ is associated with the Klimontovich microscopic phase-space density~\cite{Klimontovich_1957},
\begin{equation}\label{Klimontovich density}
\rho_K=\frac{1}{N_s}\sum_{j=0}^{N_s-1}\delta(\bol I-\bol I^j(t))\delta(\bol\theta-\bol\theta^j(t))
\end{equation}
where $\delta$ denotes the Dirac delta distribution and 
\begin{equation}\label{normalization}
\int_\mathcal{V}\rho_K(\bol I^j,\bol\theta^j)d\,\bol I d\,\bol\theta=1.
\end{equation}
The Klimontovich density provides an exact microscopic description of dynamics  and is a starting point for all further kinetic theories with applications in plasma physics~\cite{Klimontovich_1972,Dougherty_1969}.

The discrete Klimontovich ensemble $\bol \rho $ can be encoded as a quantum state in two distinct ways. A natural construction is to create a $N_s$-tensor product of quantum states $\ket{\psi^j}$ according to the KvN encoding in Eq.~\eqref{KvN state},
\begin{equation}\label{vector density separable encoding}
 \ket{\rho(t)}=\bigotimes_{j=0}^{N_s-1}\ket{\psi^j(t)},\quad \ket{\psi^j(t)}=\sum_{k=0}^{N-1}\sqrt{I_k^j(t)}e^{-i\theta_k^j(t)}\ket{k}.
\end{equation}
By Liouville's theorem~\cite{Arnold_1989}, the trajectories of the Klimontovich ensemble $\bol \rho$ evolve in parallel without intersecting in the phase-space. In addition, due to the integrability constraint, each $j$-trajectory evolves unitarily as shown in Sec.~\ref{sec:3.1}. Therefore, the ensemble evolution $t\to t+\Delta t$ based on the separable KvN encoding  has a structure akin to classical parallel computation,
\begin{equation}\label{separable evolution}
\ket{\rho(t)}\xrightarrow{\hat{U}^{\otimes N_s}(\Delta t)}\bigotimes_{j=0}^{N_s-1}\ket{\psi^j(t+\Delta t)}=\ket{\rho(t+\Delta t)}.
\end{equation}
This formulation interprets the evolution of the $N_s$ trajectories as a parallel quantum process over separable states, as illustrated in Fig.~\ref{fig:2}.
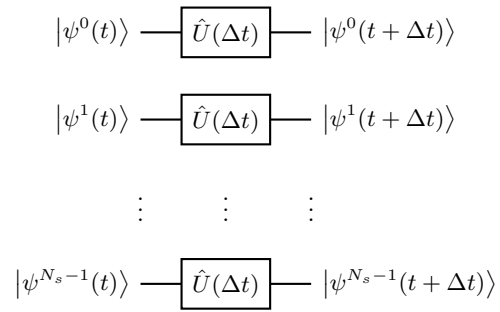
\begin{figure}[ht]
    \centering
    \begin{quantikz}
    \lstick{$\ket{\psi^0(t)}$} &\gate{\hat{U}(\Delta t)} & \qw \rstick{$\ket{\psi^0(t+\Delta t)}$} \\
    \lstick{$\ket{\psi^1(t)}$} &\gate{\hat{U}(\Delta t)} & \qw \rstick{$\ket{\psi^1(t+\Delta t)}$}\\
    \vdots&\vdots &\vdots \\
     \lstick{$\ket{\psi^{N_s-1}(t)}$} &\gate{\hat{U}(\Delta t)} & \qw \rstick{$\ket{\psi^{N_s-1}(t+\Delta t)}$}
    \end{quantikz}
    \caption{Classical-like parallel evolution of the separable quantum encoded microscopic ensemble $\ket{\rho}$ in Eq.~\eqref{separable evolution}.}
    \label{fig:2}
\end{figure}
Analyzing the depth and width of the quantum circuit in Fig.~\ref{fig:2},
\begin{itemize}
    \item Cost of implementing diagonal $\hat{U}(\Delta t)$: $\mathcal{O}(N)$,
    \item Circuit depth: $\mathcal{O}(N)$,
    \item Circuit width: $n=N_s\log_2N$ qubits.
\end{itemize}

Although the parallel implementation depicted in Fig.~\ref{fig:2} exhibits a quantum structure it manifests classical parallelism due to the separability of the input state, leading to the linear scaling of the quantum resources in terms of the ensemble size $N_s$. To leverage the quantum parallelism it is advantageous to encode the density vector $\bol \rho$ as an entangled state:
\begin{equation}\label{vector density entangled encoding}
\begin{split}
\ket{\rho(t)}&=\frac{1}{\sqrt{N_s}}\sum_{j=0}^{N_s-1}\ket{j}\ket{\psi^j(t)}\\
&=\frac{1}{\sqrt{N_s}}\sum_{j=0}^{N_s-1}\sum_{k=0}^{N-1}\sqrt{I_k^j(t)}e^{-i\theta_k^j(t)}\ket{j}\ket{k},
\end{split}
\end{equation}
which requires a total of $n=\log_2N+\log_2N_s=\log_2(N N_s)$  qubits. The unitary evolution operator $\hat{\mathcal{U}}$ now reads,
\begin{equation}\label{entangled evolution}
\hat{\mathcal{U}}(\Delta t)=\sum_j\ket{j}\bra{j}\otimes\hat{U}(\Delta t)=I_{N_s\times N_s}\otimes\hat{U}(\Delta t),
\end{equation}
and admits a straightforward implementation as shown in Fig.~\ref{fig:3}
\begin{figure}[ht]
    \centering
    \begin{quantikz}
    \lstick[4]{$\ket{\rho(t)}$} &\gate{\hat{U}(\Delta t)} &\qw \rstick[4]{$\ket{\rho(t+\Delta t)}$}\\
    &\gate{\hat{U}(\Delta t)} &\qw \\
    \vdots&\vdots &\vdots \\
     &\gate{\hat{U}(\Delta t)} &\qw 
    \end{quantikz}
    \caption{Quantum-parallel evolution of the  entangled quantum encoded microscopic density $\ket{\rho}$ through Eq.~\eqref{entangled evolution}.}
    \label{fig:3}
\end{figure}
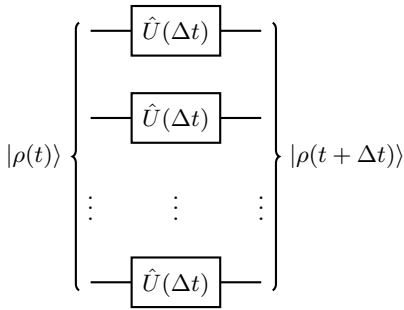

In contrast to the separable parallel implementation in Fig.~\ref{fig:2} the entangled encoding  of Eq.~\eqref{vector density entangled encoding} fully harness the quantum parallelism. By constructing a maximally entangled state the parallel evolution structure and the linear $\mathcal{O}(N)$ circuit depth are retained while an exponential reduction in the circuit width with $n=\log_2(N N_s)$ qubits is achieved. Additionally, in both parallel quantum implementations the circuit depth is independent of $N_s$. This is in sheer contrast with the classical picture where the Klimontovich  density $\rho_K$ is the starting point for kinetic-averaged descriptions but it is not useful for calculations by itself because it cannot be encoded  and tracked into a classical computer when $N_s>>1$. However, the present quantum encoding, renders the microscopic description tractable within the quantum computing framework, naturally aligning with the many-body structure of quantum systems and enabling efficient representation and evolution of large ensembles.

In various applications involving  mean-field or macroscopic descriptions, the number of canonical variables are $N\sim\mathcal{O}(1)$. For instance, in the
guiding center approximation for an axisymmetric tokamak equilibrium, there are three actions: the magnetic moment, the canonical angular momentum, and the toroidal flux enclosed by a drift surface. The respective conjugate angles are the gyrophase, azimuthal angle, and poloidal
angle~\cite{Kaufman_1972}. 

At this point it will be tempting to claim that quantum computing enables access to the evolution of large phase-space ensemble $\bol \rho$ of trajectories using quantum parallelism and superposition. However, although the evolved density vector $\ket{\rho}$ carries information about $NN_s$ variables simultaneously, this information  cannot be accessed in the same parallel way. Any measurement of an individual component $\ket{\psi^j}$ destroys the superposition,  requiring to re-execute the evolution process and therefore leading to an information extraction bottleneck associated with the quantum tomography of the vector state $\ket{\rho}$. Instead, we  have to define meaningful quantum observables that contain important information about the collective behavior of the trajectories.

\subsection{Phase-space averages as quantum observables}\label{sec:3.3}
In defining quantum observables $\hat{f}$, only those possessing a block-diagonal form in the ensemble index $j$ are physically admissible,
\begin{equation}\label{quantum observable}
\hat{f}=\sum_j\ket{j}\bra{j}\otimes\hat{f}^j,\quad \hat{f}^{j^\dagger}=\hat{f}^j,
\end{equation}
where
\begin{equation}\label{part observable}
\hat{f}^j=\sum_{k,m}^{N-1}f_{km}^j\ket{k}\bra{m},\quad f_{km}^j=f^{j*}_{mk}.
\end{equation}
This restriction is necessary because allowing non-diagonal contributions in the $\ket{j}$ basis introduces correlations and mixing between distinct  $j$-trajectories, contradicting the Liouville theorem which guarantees that trajectories evolve independently without intersection in phase-space.
Then,
\begin{widetext}
\begin{equation}\label{quantum observables continuous limit}
\begin{split}
\langle\hat{f}\rangle=\mel{\rho}{\hat{f}}{\rho}=\frac{1}{N_s}\sum_j\mel{\psi^j}{\hat{f}^j}{\psi^j}=\frac{1}{N_s}\sum_jf(\bol I^j,\bol\theta^j)=\int_\mathcal{V}f(\bol I,\bol\theta)\rho_K(\bol I, \bol \theta) d\,\bol I d\,\bol\theta \xrightarrow{N_s>>1}&\int_{\mathcal{V}}f(\bol I,\bol\theta)\langle\rho_K(\bol I, \bol \theta)\rangle d\,\bol I d\,\bol\theta\\
\overset{\langle\rho_K\rangle=\rho}{=}&\langle f(\bol I, \bol\theta)\rangle,
\end{split}
\end{equation}
\end{widetext}
where we have assumed that the Liouville smooth density $\rho(\bol I, \bol \theta,)$ is the average of the Klimontovich microscopic density $\langle\rho_K(\bol I, \bol \theta)\rangle$ for $N_s\to\infty$. Therefore, quantum observables represent averages of phase-space quantities $f(\bol\theta,\bol I)$ over the phase-space volume $\mathcal{V}$.  

In the following, we present three representative examples of observables $\hat{f}$, highlighting their structure and physical interpretation. First, consider a subset $\mathcal{V}' \subset \mathcal{V}$ containing $N_s' < N_s$ trajectories. Then,
\begin{enumerate}
    \item  Partition function for the $k$-action, $\mathcal{I}_{k,V'}$:
    \begin{equation}
    \hat{f}=\sum_{j}^{N'_s}\hat{P}_j\otimes\hat{P}_k 
    \end{equation}
     where $\hat{P}$ denotes a projection operator in the respective basis,
    \begin{equation}\label{action partition function}
    \mathcal{I}_{k,V'}=\mel{\rho}{\hat{f}}{\rho}=\frac{1}{N_s}\sum_{j}^{N'_s}I_k^j.
     \end{equation}
     \item Energy partition function, $\mathcal{E}_{V'}$:
     \begin{equation}
    \hat{f}=\sum_{j}^{N'_s}\hat{P}_j\otimes\sum_k^N\omega_k\hat{P}_k,
     \end{equation}
     \begin{equation}\label{energy partition function}
     \mathcal{E}_{V'}=\mel{\rho}{\hat{f}}{\rho}=\frac{1}{N_s}\sum_{j}^{N'_s}\sum_k^N\omega_kI^j_k.
     \end{equation}
     When $N'_s=N_s$, Eq.~\eqref{energy partition function} is the total mean energy.
     \item Coherence partition function, $\mathcal{C}_{V'}$:
     \begin{equation}
    \hat{f}=\sum_{j}^{N'_s}\hat{P}_j\otimes\sum_{k\neq m}^{N}\ket{k}\bra{m},
     \end{equation}
     \begin{equation}
    \mathcal{C}_{V'}=\mel{\rho}{\hat{f}}{\rho}=\frac{1}{N_s}\sum_{j}^{N'_s}\sum_{k\neq m}^N\sqrt{I_k^jI_m^j}\cos{(\theta_k-\theta_m)}.
     \end{equation}
     The last expression provides information about the cross correlations between the various modes.
\end{enumerate}

The key advantage of encoding of phase-space averages as quantum observables lies in the fundamental difference between quantum and classical sampling with respect to the ensemble size $N_s$. In classical computation, estimation of the observable $\langle f(\bol I,\bol\theta)\rangle$ to an error $\varepsilon$ requires $\mathcal{O}(N_s)$ samples. In contrast, the quantum counterpart, based on the quantum amplitude estimation~\cite{Brassard_2002,Grikno_2021} routine, achieves the same  accuracy with only $\mathcal{O}(\sqrt{N_s})$, yielding a quadratic speed-up. In that direction, observable extraction can be framed within a quantum Monte Carlo paradigm, as proposed and demonstrated in~\cite{Gamiz_2026}.

All together, the present framework introduces two quantum improvements into the classical action-angle description of integrable Hamiltonian systems. Firstly, the entangled, finite dimensional KvN encoding allows description and parallel evolution of $N_s$ trajectories in the phase-space with only $\log_2N_s$ qubits, achieving an exponential quantum compression in the ensemble size. Secondly, extracting collective phase-space quantities as quantum observables benefits from quantum sampling speed-ups, reducing the measurement cost from $\mathcal{O}(N_s)$ to $\mathcal{O}(\sqrt{N_s})$ for fixed accuracy $\varepsilon$. Once again, the unitary nature of the quantum routines ensures the preservation of the symplectic structure for finite-dimensional Hamiltonian integrable systems.

\section{Beyond integrability: Leveraging the canonical perturbation techniques}\label{sec:4}

Let us now drift away from integrable systems and examine classical non-integrable and nonlinear Hamiltonian systems of the form,
\begin{equation}\label{non-integrable Hamiltonian}
H(\bol \theta,\bol I, t)=H_0(\bol I)+\epsilon H_1(\bol\theta. \bol I, t).
\end{equation}
The Hamiltonian~\eqref{non-integrable Hamiltonian} consists of an integrable part $H_0$  and a non-integrable contribution $H_1$ whose strength is modulated by the dimensionless ordering parameter $\epsilon$. When $\epsilon<<1$ the non-integrable contribution acts as a perturbation to the integrable part and the Hamiltonian system is described as near-integrable.

Hamiltonian systems of the form~\eqref{non-integrable Hamiltonian} play a central role in  classical dynamics with a wide range of practical applications such as wave-particle interactions in magnetized plasmas, galactic dynamics as well as  the theory of the onset of Hamiltonian chaos~\cite{Lichtenberg_1992}.

Finding the Hamiltonian flow $\bol z(t)=\Phi_t[\bol z_0]$ for the $H$ in Eq.~\eqref{non-integrable Hamiltonian} requires solving the Hamilton's equations,
\begin{equation}\label{Hamiltons equation and flow}
\dv{\bol z}{t}=\Omega\nabla_z H
\end{equation}
with initial condition $\bol{z}_0$. However, these equations are generally nonlinear and their explicit symplectic integration can be computationally hard, especially when the non-integrable part $H_1$ is non-separable in the phase-space variables $\bol z$~\cite{He_2015}. To address this, we focus on explicit near-symplectic integration based on the canonical perturbation theory~\cite{Kominis_2008}. For this task, we seek a nonlinear canonical transformation $T$ mediated by the generating function $W$ that depends on the ordering parameter $\epsilon$ such that,
\begin{equation}\label{integrable perturped hamiltonian}
T:\bol z\to\Bar{\bol z}(\bol z, \epsilon),\quad H(\bol z, t)\to K(\Bar{\bol I}).
\end{equation}
Note that the transformed Hamiltonian $K$ in Eq.~\eqref{integrable perturped hamiltonian} depends only on the new actions $\bar{\bol I}$, and is therefore integrable system  up to order $\epsilon^\kappa$, where $\kappa\geq1$ is the order of the $T$ canonical transformation. As a result, in accordance to Eqs.~\eqref{action-angle evolution},~\eqref{action-angle evolution2},
\begin{align}\label{approximate evolution}
& \Bar{I}_k(t)=\Bar{I}_k(t_0)+\mathcal{O}(\epsilon^\kappa),\\
&\Bar{\theta}_k=\Bar{\theta}_k(t_0)+\Bar{\omega}_k(t-t_0)+\mathcal{O}(\epsilon^\kappa),\,\,\Bar{\omega}_k=\pdv{K}{\Bar{I}_k}.
\end{align}

Consequently, by constructing the nonlinear canonical transformation $T$, the framework developed in Sec.~\ref{sec:3} is directly amenable under the approximation ordering $\mathcal{O}(\epsilon^\kappa)$. In the following, we present the Lie canonical perturbation theory and the toolkit for obtaining $T$  as well as discuss its use as an efficient near-symplectic integrator for long-time dynamics, and outline its inherent limitations. Finally we explicitly calculate the expression of the general quantum observable in Eq.~\eqref{quantum observable} reflecting the collective dynamics for the non-integrable system $H$.

\subsection{Lie canonical perturbation theory}\label{sec:4.1}
Given the Hamiltonian in form~\eqref{non-integrable Hamiltonian} the  Lie theory provides an explicit canonical perturbation scheme with versatile applications~\cite{Kominis_2008, Lichtenberg_1992, Cary_1981, Kominis_2010, Kominis_2012}. In the Lie canonical perturbation theory the canonical transformation $T$ is expressed through the Lie operator $L=\{W, \,\,\}$ as,
\begin{equation}\label{T Lie transformation}
T=e^{-L},
\end{equation}
where $W$ is the Lie generating function and $\{\cdot,\cdot\}$ denotes the Poisson bracket. Expressing the participating quantities $X=\{H, K, T, W, L\}$  as  power series in terms of $\epsilon$
\begin{equation}\label{expansion}
X(\bol z, t, \epsilon)=\sum_{\kappa=0}^\infty X_\kappa(\bol z ,t)\epsilon^\kappa,
\end{equation}
 and substituting into Eq.~\eqref{T Lie transformation} the relation between the $T_\kappa$ and $L_\kappa$ quantities is obtained, given that $T_0=I_{N\times N}$, Although $T$ is a canonical transformation, each $T_\kappa$ is a near symplectic--contact--transformation. Then, to second order in $\epsilon$, the transformations $T$ and $T^{-1}$ are given by,
 \begin{equation}\label{T_1 and T_2}
T_1=-L_1,\quad T_2=\frac{1}{2}(L_1^2-L_2)
 \end{equation}
and
\begin{equation}\label{inverse T_1 and T_2 }
 T_1^{-1}=L_1,\quad T_2^{-1}=\frac{1}{2}(L_1^2+L_2).
\end{equation}

The respective Lie generating functions $W_{1,2}$ satisfy: 
\begin{align}
&\pdv{W_1}{t}+\{W_1, H_0\}=K_1-H_1, \label{W_1}\\
&\pdv{W_2}{t}+\{W_2, H_0\}=2K_2-L_1(K_1+H_1).\label{W_2}
\end{align}
The left hand sides in both Eqs.~\eqref{W_1},~\eqref{W_2} are the total derivative $D_0W_{1,2}$ along the unperturbed orbits of $H_0$. Selecting $K_1=K_2=0$, the new Hamiltonian $K$ to order $\mathcal{O}(\epsilon^2)$ is,
\begin{equation}\label{K to second order}
K(\Bar{\bol I})=H_0(\Bar{\bol I}).
\end{equation}
The generating function $W_1$ is derived from Eq.~\eqref{W_1} as,
\begin{equation}\label{W_1 derivation}
W_1=-\int_{ \substack{ \bol I=\text{constant}\\
\bol\theta (t)=\bol\theta(t_0)+\bol\omega(t-t_0)}}^t H_1(\bol I, \bol\theta, s)d\,s.
\end{equation}
Given the form of $H_1$, the generating function $W_1$ can be obtained from the integral in Eq.~\eqref{W_1 derivation} either analytically or numerically. In either case, the key point here is the existence of a near symplectic transformation $T_1$  that maps the Hamiltonian $H$ in Eq.~\eqref{non-integrable Hamiltonian} into a new Hamiltonian $K$ in Eq.~\eqref{K to second order} which now generates unitary flow to order $\epsilon$.  From Eq.~\eqref{T_1 and T_2}, the new transformed variables read to first order:
\begin{align}
&\Bar{I}_k=I_k-\epsilon\pdv{W_1}{\theta_k},\label{transformed new action}\\
&\Bar{\theta}_k=\theta_k+\epsilon\pdv{W_1}{I_k}. \label{transformed new angle}
\end{align}

Consequently, rather than directly simulating the nonlinear Hamilton's equations Eq.~\eqref{Hamiltons equation and flow} associated with Eq.~\eqref{non-integrable Hamiltonian}, one can systematically construct nonlinear and near symplectic transformations $T_\kappa$, generated by the Lie functions $W_\kappa$, to recast the system into an approximately integrable form to a desired order. This transformation exposes an underlying (approximate) unitary flow in phase-space as discussed in Sec.~\ref{sec:3.1}. In the following section we discuss the validity of this unitary approximation and whether it can be used for long time  quantum symplectic evolution.

\subsection{Validity of the unitary evolution approximation and complexity considerations}\label{sec:4.2}

In Sec.~\ref{sec:3.1}, we demonstrated that the Hamiltonian flow of an integrable system, expressed in the action–angle variables is unitary, leading to the evolution operators $\hat{U}^{\otimes N_s}(\Delta t)$ and $\hat{\mathcal{U}}(\Delta t)$. This correspondence is exact and remains valid for arbitrarily long simulation times $\Delta t = t - t_0$.

For the Lie perturbation method to be applicable, the transformation $T_\kappa$ is required to remain well-defined and bounded. This introduces the well-known problem of small denominators: as the perturbation strength $\epsilon$ increases, resonant surfaces overlap, leading to the onset of chaotic behavior followed by the breakdown of the perturbative expansion and the approximate invariants in Eq.~\eqref{approximate evolution} over long times. On the other hand, under non-resonant conditions and for sufficiently small $\epsilon$, the invariant tori persist, although deformed, according to the Kolmogorov–Arnold–Moser (KAM) theorem~\cite{Arnold_1989}.

To make things explicit, let us assume an one-dimensional, time-periodic perturbation $H_1$ in Eq.~\eqref{non-integrable Hamiltonian}. Expanding in Fourier modes over the angle and time,
\begin{equation}\label{Fourier of H_1}
H_1(\theta, I, t)=\sum_{m\neq0}^\infty A_m(I)e^{i(m\theta(t)-\omega_m t)}
\end{equation}
the first order generating function $W_1$ is calculated along the unperturbed orbits via the integral in Eq.~\eqref{W_1 derivation},
\begin{equation}\label{spesific W_1}
W_1=-\sum_{m\neq0}^\infty A_m(I)e^{i(m\theta_0-\omega_mt_0)}\Big[\frac{e^{i\Omega_m t}-e^{i\Omega_m t_0}}{i\Omega_m}\Big],
\end{equation}
where $\Omega_m=m\omega_0-\omega_m$  and $\omega_0=\partial H_0(I)/\partial I$ is the unperturbed frequency.

The behavior of $W_1$ in Eq.~\eqref{spesific W_1} is governed by the resonance condition $\Omega_m = 0$. In the long-time limit $t \to \infty$,
\begin{equation}\label{resonance blow up}
\lim_{t\to\infty}W_1=2\pi\delta(\Omega_m).
\end{equation}
signaling the breakdown of the perturbative expansion at resonance. However, for any finite time interval $\Delta t=t-t_0$ the limit remains regular,
\begin{equation}\label{reonance no blow}
\lim_{\Omega_m\to0}W_1=-(t-t_0)\sum_{m\neq0}^\infty A_m(I)e^{i(m\theta_0-\omega_mt_0)}.
\end{equation}
demonstrating that $W_1$ is bounded even at resonance~\cite{Kominis_2008, Kominis_2010,Abdullaev_2002}.

To highlight the interplay between the required the smallness of the time step $\Delta t$ with the perturbation strength $\epsilon$, assume that $H_1$ is $p$-times differentiable so its Fourier coefficients decay as $\abs{A_m} = \mathcal{O}(\abs{m}^{-(p+1)})$, yielding the bound
\begin{equation}\label{bound}
\abs{\lim_{\Omega_m\to0}W_1}\leq\Delta t\sum_{m\neq0}^\infty \frac{1}{\abs{m}^{p+1}}=2\Delta t\zeta(p+1),
\end{equation}
where $\zeta$ is the Riemann zeta function. Equation~\eqref{bound} reveals that the validity of the unitary evolution approximation under the Lie perturbation is controlled not only by the perturbation strength $\epsilon$, but by an effective $\epsilon'$,
\begin{equation}\label{new epsilon}
\epsilon'\sim2\epsilon\Delta t\zeta(p+1).
\end{equation}
This condition allows us to maintain $\epsilon'<<1$  even in regimes where $\epsilon$ is not small, such as near the breakdown of KAM tori, by selecting sufficiently small time steps $\Delta t$. However, this comes at a computational cost: reducing $\Delta t$ increases the number of required evolution steps $N_t$ for a total simulation time $\mathcal{T} = N_t \Delta t$. 

In the transformed--$\epsilon$-integrable--representation  $(\Bar{\bol\theta}, \Bar{\bol I})$ the quantum encoding is according to Eq.~\eqref{vector density entangled encoding},
\begin{equation}\label{transformed entangled density}
 \ket{\Bar{\rho}}= \frac{1}{\sqrt{N_s}}\sum_{j=0}^{N_s-1}\sum_{k=0}^{N-1}\sqrt{\Bar{I}_k^j(t)}e^{-i\Bar{\theta}_k^j(t)}\ket{j}\ket{k}.
\end{equation}
Evolving this state over a total time $\mathcal{T}$ requires $N_t$ sequential applications of the unitary blocks of $\hat{U}(\Delta t)$ in Eq.~\eqref{unitary diagonal integrable} (under the substitution $\omega_k\to\Bar{\omega}_k=\pdv{K}{\Bar{I}_k}$), resulting in an overall circuit depth scaling as $\mathcal{O}(N_t N)$, as presented in Fig.~\ref{fig:4}. Given  the relation for the method's approximation error $\epsilon'$ for a time-step $\Delta t$  in Eq,~\eqref{new epsilon}, the total error $\epsilon_{t}$ after $N_t$ steps of unitary and symplectic evolution is therefore $\epsilon_t=N_t\epsilon'\sim\epsilon\mathcal{T}$.
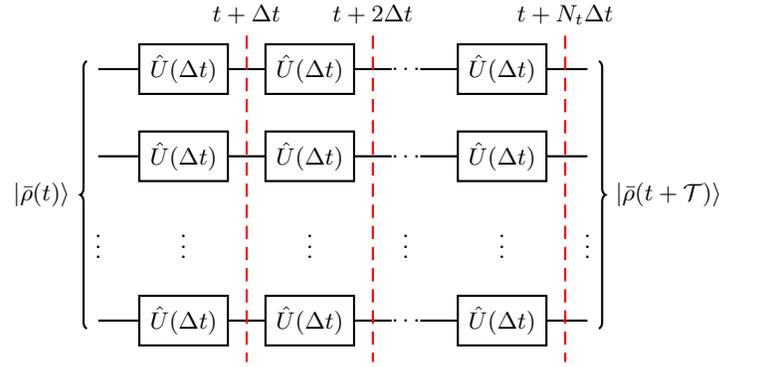
\begin{figure}[ht]
    \centering
   \begin{quantikz}
    \lstick[4]{$\ket{\Bar{\rho}(t)}$} &\gate{\hat{U}(\Delta t)}\slice{$t+\Delta t$} &\gate{\hat{U}(\Delta t)}\slice{$t+2\Delta t$} &\qw \hdots &\gate{\hat{U}(\Delta t)}\slice{$t+N_t \Delta t$}&\qw \rstick[4]{$\ket{\Bar{\rho}(t+\mathcal{T})}$}\\
    &\gate{\hat{U}(\Delta t)} &\gate{\hat{U}(\Delta t)}&\qw \hdots &\gate{\hat{U}(\Delta t)}&\qw\\
    \vdots&\vdots &\vdots&\vdots&\vdots&\vdots \\
     &\gate{\hat{U}(\Delta t)} &\gate{\hat{U}(\Delta t)} &\qw\hdots &\gate{\hat{U}(\Delta t)}&\qw
    \end{quantikz}
    \caption{Unitary and symplectic evolution of the state $\ket{\Bar{\rho}}$ in Eq.~\eqref{transformed entangled density} for a total time $\mathcal{T}=N_t\Delta t$. }
    \label{fig:4}
\end{figure}

In general, for higher integration precision one can leverage the explicit structure of the Lie perturbation method and proceed into calculating the second order Lie generating function from Eq.~\eqref{W_2}. Carrying out higher-order Lie expansions leads to an effective error per time step $\epsilon'$ of the form~\cite{Kominis_2008},
\begin{equation}\label{new new epsilon}
\epsilon'\sim\epsilon\Delta t^\nu, \quad \nu\geq1
\end{equation}
where $\nu$ depends on the order of the Lie expansion. For the presented first order scheme,  $\nu=1$ according to Eq.~\eqref{new epsilon}. As a result,  the complexity of the higher order unitary-symplectic evolution $t\to t+\mathcal{T}$ is,
\begin{equation}\label{scaling of quantum}
\mathcal{O}\Big[N\Big(\frac{\epsilon\mathcal{T}^\nu}{\epsilon_t}\Big)^{1/(\nu-1)}\Big],\quad \epsilon\in[0,1].
\end{equation}
Taking into consideration the $\sqrt{N_s}$ Grover overhead for reading the  $\mel{\Bar{\rho}}{\hat{f}}{\Bar{\rho}}$ observables in the $(\Bar{\bol\theta}, \Bar{\bol I})$ representation, in accordance with Sec.~\ref{sec:3.3}, the overall complexity of the proposed framework is,
\begin{equation}\label{quantum complexity}
\mathcal{O}\Big[\log_2(N N_s),\,\,\sqrt{N_s}N\Big(\frac{\epsilon\mathcal{T}^\nu}{\epsilon_t}\Big)^{1/(\nu-1)}\Big].
\end{equation}
For comparison, the corresponding classical  $\kappa$-order symplectic integration and extraction of the phase space quantities scale as~\cite{Feng_2010},
\begin{equation}\label{classica complexity}
\mathcal{O}\Big[2N N_s,\,\,N_s poly(N)\frac{\mathcal{T}}{\epsilon_t^{1/\kappa
}}\Big)\Big].
\end{equation}

The comparison between the scaling of the quantum workflow in Eq.~\eqref{quantum complexity} and the classical symplectic integration in Eq.~\eqref{classica complexity} highlights two key advantages. The proposed quantum framework achieves an exponential compression in memory resources, reducing the representation of the full phase-space ensemble from $\mathcal{O}(NN_s)$ classical degrees of freedom to $\mathcal{O}(\log_2(NN_s))$ qubits. In addition, it showcases a potential polynomial speed-up in terms of  the problem size  $NN_s$, arising from both the quantum-parallel evolution and the quadratic improvement in sampling complexity. Finally, for higher order Lie perturbation schemes ($\nu>>1$) we approach the linear scaling with respect to the total simulation time $\mathcal{T}$.

The proposed framework transforms nonlinear Hamiltonian dynamics into an inherently finite-dimensional unitary evolution through the classical construction of action–angle variables and the nonlinear contact transformation $T_\kappa$, after which the described quantum techniques apply seamlessly. Regarding the dynamical regime for which this method is faithful, we demonstrated via Eq.~\eqref{reonance no blow} that the  finite time-step Lie perturbation theory is valid at resonances with a well-controlled error $\epsilon'$. However, both the action–angle transformation and the Lie transformation $T_\kappa$ are local constructions, defined on invariant tori. The global error after a total integration time $\mathcal{T}$ is given by the deviation between the exact flow $\Phi_\mathcal{T}^H[\bol z_0]$, generated by the Hamiltonian $H$, and the approximate Lie flow $\Phi_\mathcal{T}^L[\bol z_0]$~\cite{Hairer_2006},
\begin{equation}\label{global error}
\norm{\Phi_\mathcal{T}^H[\bol z_0]-\Phi_\mathcal{T}^L[\bol z_0]}\leq C(\epsilon)\frac{e^{\lambda\mathcal{T}}-1}{\lambda},\quad C(\epsilon)>0.
\end{equation}
In Eq.~\eqref{global error}, $\lambda$ is the finite-time maximal Lyapunov exponent~\cite{Arnold_1989,Lichtenberg_1992,Saletan_1998,Hairer_2006} characterizing the local stretching rate of nearby trajectories of $H$, and $C(\epsilon)$ is a polynomial function of the perturbation strength $\epsilon$. For strongly nonlinear systems with $\epsilon\to1$, resonances overlap giving rise to gaps and chaotic layers in which KAM tori break down. In these chaotic regions with $\lambda>0$, the bound in Eq.~\eqref{global error} scales exponentially with the total evolution time $\mathcal{T}$, indicating that the Lie method ceases to faithfully represent the underlying nonlinear dynamics. In contrast, within surviving KAM tori, where $\lambda\to0$, the bound reduces to linear scaling in $\mathcal{T}$, ensuring reliable Lie-based integration.

Therefore, our quantum-symplectic framework is well-suited for simulating regular nonlinear Hamiltonian dynamics on persistent KAM invariant tori, but is not applicable in chaotic regions of phase-space.\\

\subsection{Evolution of quantum observables}\label{sec:4.3}
In Sec.~\ref{sec:3.3} we demonstrated that quantum observables $\mel{\rho}{\hat{f}}{\rho}$ correspond to phase-space averages $\langle f(\bol I, \bol\theta)\rangle$ in the continuous limit, providing some important quantities for integrable systems. We now examine the general form and the evolution of these quantum observables when transformed from the measurement $(\Bar{\bol\theta},\Bar{\bol I})$ frame back into the original $({\bol\theta},{\bol I})$ coordinates.

Given the general form of a time-independent quantum observable operator $\hat{f}$ in Eqs.~\eqref{quantum observable},~\eqref{part observable}, the mean value observable $\mel{\Bar{\rho}}{\hat{f}}{\Bar{\rho}}$ reads,
\begin{equation}\label{expicit observable transformed}
\Bar{f}(\Bar{\bol I},\Bar{\bol\theta})=\mel{\Bar{\rho}}{\hat{f}}{\Bar{\rho}}=\sum_{j=0}^{N_s-1}\sum_{k,m=0}^{N-1}f_{km}^j\sqrt{\Bar{I}^j_m \Bar{I}_k^j}e^{-i(\Bar{\theta}^j_m-\Bar{\theta}_k^j)}.
\end{equation}
Inserting the transformation relations~\eqref{transformed new action},~\eqref{transformed new angle} into the observable of Eq.~\eqref{expicit observable transformed} yields, to order $\epsilon$,
\begin{widetext}
\begin{equation}\label{tranformed general observable}
f(\bol I,\bol\theta,t)=\sum_{j=0}^{N_s-1}\sum_{k,m=0}^{N-1}f_{km}^j\sqrt{I^j_k I^j_m} e^{-i(\theta_m^j-\theta_k^j)}\Big\{1-\epsilon\Big[\frac{1}{2 I_k^j I_m^j}\Big(I_k^j\pdv{W_1}{\theta^j_m}+I_m^j\pdv{W_1}{\theta^j_k}\Big)+i\Big(\pdv{W_1}{I_m^j}-\pdv{W_1}{I_k^j}\Big)\Big]\Big\}.
\end{equation}
\end{widetext}
The quantum observable in Eq.~\eqref{tranformed general observable} describes the first order behavior of a general phase-space quantity $f(\bol I,\bol\theta, t)$ due to the contribution of the non-integrable part $H_1$. Notably, Eq.~\eqref{tranformed general observable} contains nonlinear contributions from the first order Lie generating function $W_1$ in Eq.~\eqref{W_1 derivation}, reflecting the nonlinear deformation of the underlying integrable dynamics.

We proceed by explicitly calculating the classical evolution equation of the quantum observable~\eqref{tranformed general observable} to the same approximation. Starting from Eq.~\eqref{expicit observable transformed} the total derivative is,
\begin{equation}\label{untransformed derivative}
\dv{\Bar{f}}{t}=-i\sum_{j}\sum_{k,m}(\Bar{\omega}_m^j-\Bar{\omega}_k^j)\Bar{F}_{km}^j,
\end{equation}
where the $\Bar{F}_{km}^j$ is the summand in Eq.~\eqref{expicit observable transformed}.
Similarly, the summands in Eq.~\eqref{tranformed general observable} to the respective order in $\epsilon$ are denoted as $F_{0km}^j$ and $F_{1km}^j$ such that  $\Bar{F}_{km}^j=F_{0km}^j-\epsilon F_{1km}^j$. As a result, to evaluate expression~\eqref{untransformed derivative} into the initial coordinates we need to find the transformation $\Bar{\omega}_k\to\omega_k$. In that respect,
\begin{equation}\label{chain rule}
\Bar{\omega}_k=\pdv{K}{\Bar{I}_k}=\pdv{K}{I_k}\pdv{I_k}{\bar{I}_k}
\end{equation}
and from Eq.~\eqref{K to second order} and Eq.~\eqref{transformed new action},
\begin{equation}\label{K in terms of initial variables}
K=H_0(I_k)-\epsilon\pdv{W_1}{\theta_k}\pdv{H_0}{I_k}.
\end{equation}
Reminding that $\omega_{0k}=\pdv{H_0}{I_k}$ and using Eq.~\eqref{K in terms of initial variables} and Eq.~\eqref{transformed new action}, the $\Bar{\omega}_k$ in Eq.~\eqref{chain rule} obtains, in terms of the $(\bol I, \bol\theta)$, the form
\begin{equation}\label{omega transformation}
\Bar{\omega}_k=\omega_{0k}-\epsilon\pdv{W_1}{\theta_k}\pdv{\omega_{0k}}{I_k}.
\end{equation}
Inserting Eq.~\eqref{omega transformation} and the summand transformation  $\Bar{F}_{km}^j=F_{0km}^j-\epsilon F_{1km}^j$ into Eq.~\eqref{untransformed derivative}, the evolution equation for the quantum observable is finally obtained (for time-independent $W_1$),
\begin{widetext}
\begin{equation}\label{evolution observable equation}
\dv{f}{t}=-i\sum_j\sum_{km}(\omega_{0m}-\omega_{0k})F_{0km}^j+i\epsilon\sum_j\sum_{km}(\omega_{0m}-\omega_{0k})F^j_{1km}+i\epsilon\sum_j\sum_{km}
\Big(\pdv{W_1}{\theta_m^j}\pdv{\omega_{0k}}{I^j_k}-\pdv{W_1}{\theta^j_k}\pdv{\omega_{0k}}{I_k^j}\Big)F^j_{1km}.
\end{equation}
\end{widetext}

The contributions from the evolution equation~\eqref{evolution observable equation} are analyzed as follows:.
\begin{itemize}
    \item The term $(\omega_{0m}-\omega_{0k})$ reflects the phase mixing between the unperturbed frequencies and contributes both to the zero order unperturbed part $F_{0km}^j$ and the first order $F_{1km}^j$.
    \item The first order term $F_{1km}^j$ contains derivatives $\partial_I W_1$ and $\partial_\theta W_1$ of the nonlinear generating function $W_1$, signifying nonlinear deformations of the invariant tori.
    \item Finally, the term in the parenthesis in Eq.~\eqref{evolution observable equation} contains the frequency gradients $\partial_I\omega$, representing frequency shifts and mixing of the KAM tori~\cite{Arnold_1989}.
\end{itemize}
In summary, Eq.~\eqref{evolution observable equation} describes the  transport of phase-space observables along the KAM invariant tori~\cite{Kominis_2010,Kominis_2012}, induced by combining classical nonlinear near-symplectic transformations $T$, parallel unitary quantum evolution of entangled trajectories, and quantum computing calculations for large ensemble of the discrete phase-space trajectories.

\section{Conclusions}\label{Conclusions}
Simulating nonlinear dynamics with quantum computers is a challenging task due to the inherent linearity of quantum evolution. In this work, we bridge this gap by establishing rigorous connections between quantum computing techniques and the classical theory of  finite-dimensional Hamiltonian systems. In particular, we develop a unified quantum-symplectic framework  exploiting the intrinsic geometric connection of unitary quantum dynamics with specific symplectic phase-space dynamics.

Our starting point is a purely geometric description of quantum mechanics into a real Kähler space, where the symplectic structure is explicitly manifested. Within this representation exact connections between unitary quantum Schrodinger evolution and classical quadratic Hamiltonian flows are identified. This geometric picture allows us to identify which quadratic Hamiltonian systems can be recast into the quantum Schrodinger form of Eq.~\eqref{quantal Schrodinger}, and thus have the potential to be quantum simulated exponentially fast, as demonstrated in~\cite{Babbush_2023}.

Extending beyond the quadratic forms, we showcase that Liouville integrability provides an explicit path to unitary evolution under the action-angle transformations. By encoding the action-angle variables into a KvN quantum state we obtain  an exact finite-dimensional quantum representation due to the integrability constraint. This sets up the basis for parallelized quantum evolution of trajectories ensemble in the phase space using the Liouville theorem. In addition, we show how quantum observables can be constructed from the ensemble of phase-space trajectories and that quantum sampling converges to classical averages in the continuous limit.

As a next step. we tackle non-integrable and  nonlinear Hamiltonian systems by employing the Lie perturbation theory as a systematic tool for constructing near-integrable representations. This approach enables approximate unitary evolution over finite time scales with controllable error governed by the perturbation strength and time discretization step. Notably, the quantum processes including  unitary evolution and observable measurement propose  exponential resource savings and polynomial speed-up with respect to the total number of phase-space ensemble variables, when compared to the standard symplectic integration techniques.

Finally, we derive the time evolution equation for the quantum observables in  phase-space. The nonlinear evolution equation~\eqref{evolution observable equation} captures essential phenomena of nonlinear Hamiltonian flows such as deformation, transport, and frequency mixing in the phase space, illustrating how classical nonlinear transformations can unlock practical applications for quantum computing.

Although the proposed unitary-symplectic framework is applicable to regular nonlinear Hamiltonian dynamics, the broader objective of this work is to establish a coherent framework connecting classical Hamiltonian mechanics, symplectic geometry, and quantum computation. Unlike trajectory-level quantum simulation techniques, our approach properly blends quantum algorithmic modules and ideas with the classical theory of Hamiltonian systems providing a foundation for future structure-preserving quantum simulation in high-dimensional classical systems.

\section*{Acknowledgments}
This work has been carried out within the framework of the EUROfusion Consortium, funded by the European Union via the Euratom Research and Training Programme (Grant Agreement No 101052200 — EUROfusion). Views and opinions expressed are however those of the authors only and do not necessarily reflect those of the European Union or the European Commission. Neither the European Union nor the European Commission can be held responsible for them. A.K.R. and G.V. are supported by the US 
Grant Nos. DE-SC0021647, DE-FG02-91ER-54109 DE-SC0021651.
E.K. acknowledges Yannis Kominis for introducing him to Lie perturbation theory

\textit{This work is dedicated to the memory of Prof. Nuno Loureiro}.

\bibliography{ref}

@PREAMBLE{
 "\providecommand{\noopsort}[1]{}" 
 # "\providecommand{\singleletter}[1]{#1}%" 
}

@article{Childs_2012,
volume={12},
author={Andrew M. Childs and Nathan Wiebe},
title={ Hamiltonian simulation using linear combinations of unitary operations},
   doi={10.26421/QIC12.11-12-1},
   journal={Quantum Inf. Comput.},
   year={2012},
}

@article{Schlimgen_2021,
  title = {Quantum Simulation of Open Quantum Systems Using a Unitary Decomposition of Operators},
  author = {Schlimgen, Anthony W. and Head-Marsden, Kade and Sager, LeeAnn M. and Narang, Prineha and Mazziotti, David A.},
  journal = {Phys. Rev. Lett.},
  volume = {127},
  issue = {27},
  pages = {270503},
  numpages = {6},
  year = {2021},
  doi = {10.1103/PhysRevLett.127.270503}
}

@article{Schlimgen_2022,
  title = {Quantum state preparation and nonunitary evolution with diagonal operators},
  author = {Schlimgen, Anthony W. and Head-Marsden, Kade and Sager-Smith, LeeAnn M. and Narang, Prineha and Mazziotti, David A.},
  journal = {Phys. Rev. A},
  volume = {106},
  issue = {2},
  pages = {022414},
  year = {2022},
  doi = {10.1103/PhysRevA.106.022414}
}

@article{Jin_2023,
  title = {Quantum simulation of partial differential equations: Applications and detailed analysis},
  author = {Jin, Shi and Liu, Nana and Yu, Yue},
  journal = {Phys. Rev. A},
  volume = {108},
  issue = {3},
  pages = {032603},
  numpages = {20},
  year = {2023},
  doi = {10.1103/PhysRevA.108.032603},
}

@article{Koukoutsis_2025,
author = {Koukoutsis, Efstratios and Papagiannis, Panagiotis and Hizanidis, Kyriakos and Ram, Abhay K. and Vahala, George and Amaro, {\'O}scar and Gamiz, Lucas I Iñigo and Vallis, Dimosthenis},
doi = {10.2478/qic-2025-0007},
title = {Quantum Implementation of Non-unitary Operations with Biorthogonal Representations},
journal = {Quantum Inf. Comput.},
number = {2},
volume = {25},
year = {2025},
pages = {141--155}
}

@article{Amaro_2025, 
title={Variational quantum simulation of the Fokker–Planck equation applied to quantum radiation reaction},
volume={91}, 
DOI={10.1017/S0022377825100652},
number={4}, 
journal={J. Plasma Phys.},
author={Amaro, {\'O}scar and Iñigo Gamiz, Lucas Ivan and Vranic, Marija},
year={2025}, 
pages={122}
}

@article{Koukoutsis_2023,
  title = {Dyson maps and unitary evolution for Maxwell equations in tensor dielectric media},
  author = {Koukoutsis, Efstratios and Hizanidis, Kyriakos and Ram, Abhay K. and Vahala, George},
  journal = {Phys. Rev. A},
  volume = {107},
  issue = {4},
  pages = {042215},
  numpages = {10},
  year = {2023},
  doi = {10.1103/PhysRevA.107.042215},
}

@article{Bosch_2025,
  title = {Quantum wave simulation with sources and loss functions},
  author = {B{\"o}sch, Cyrill and Schade, Malte and Aloisi, Giacomo and Keating, Scott D. and Fichtner, Andreas},
  journal = {Phys. Rev. Res.},
  volume = {7},
  issue = {3},
  pages = {033225},
  numpages = {18},
  year = {2025},
  month = {Sep},
  doi = {10.1103/t8wt-mp2l},
}

@misc{Gamiz_2026,
      title={Quantum Monte Carlo Simulations for predicting electron-positron pair production via the linear Breit-Wheeler process}, 
      author={Lucas I. Iñigo Gamiz and Óscar Amaro and Efstratios Koukoutsis and Marija Vranić},
      year={2026},
      eprint={2601.03953},
      archivePrefix={arXiv},
      primaryClass={physics.plasm-ph},
      url={https://arxiv.org/abs/2601.03953}, 
}

@article{Babbush_2023,
  title = {Exponential Quantum Speedup in Simulating Coupled Classical Oscillators},
  author = {Babbush, Ryan and Berry, Dominic W. and Kothari, Robin and Somma, Rolando D. and Wiebe, Nathan},
  journal = {Phys. Rev. X},
  volume = {13},
  issue = {4},
  pages = {041041},
  numpages = {34},
  year = {2023},
  doi = {10.1103/PhysRevX.13.041041},
}

@article{Bhuvanesh_2026,
  title = {Simulating plasma wave propagation on a superconducting quantum chip},
  author = {Sundar, Bhuvanesh and Evert, Bram and Geyko, Vasily and Patterson, Andrew and Joseph, Ilon and Shi, Yuan},
  journal = {Phys. Rev. Appl.},
  volume = {25},
  issue = {2},
  pages = {024077},
  numpages = {24},
  year = {2026},
  month = {Feb},
  doi = {10.1103/lxr6-t7vb},
}

@article{An_2023,
  title = {Linear Combination of Hamiltonian Simulation for Nonunitary Dynamics with Optimal State Preparation Cost},
  author = {An, Dong and Liu, Jin-Peng and Lin, Lin},
  journal = {Phys. Rev. Lett.},
  volume = {131},
  issue = {15},
  pages = {150603},
  numpages = {6},
  year = {2023},
  doi = {10.1103/PhysRevLett.131.150603},
}

@article{Wawrzyniak_2025,
title = {A quantum algorithm for the lattice-Boltzmann method advection-diffusion equation},
journal = {Comput. Phys. Commun.},
volume = {306},
pages = {109373},
year = {2025},
issn = {0010-4655},
doi = {https://doi.org/10.1016/j.cpc.2024.109373},
author = {David Wawrzyniak and Josef Winter and Steffen Schmidt and Thomas Indinger and Christian F. Jan{\ss}en and Uwe Schramm and Nikolaus A. Adams}
}

@article{Novikau_2023,
  title = {Simulation of Linear Non-Hermitian Boundary-Value Problems with Quantum Singular-Value Transformation},
  author = {Novikau, I. and Dodin, I.Y. and Startsev, E.A.},
  journal = {Phys. Rev. Appl.},
  volume = {19},
  issue = {5},
  pages = {054012},
  numpages = {23},
  year = {2023},
  doi = {10.1103/PhysRevApplied.19.054012},
}

@article{Harrow_2009,
  title = {Quantum Algorithm for Linear Systems of Equations},
  author = {Harrow, Aram W. and Hassidim, Avinatan and Lloyd, Seth},
  journal = {Phys. Rev. Lett.},
  volume = {103},
  issue = {15},
  pages = {150502},
  numpages = {4},
  year = {2009},
  month = {Oct},
  publisher = {American Physical Society},
  doi = {10.1103/PhysRevLett.103.150502},
  url = {https://link.aps.org/doi/10.1103/PhysRevLett.103.150502}
}

@article{Berry_2014,
doi = {10.1088/1751-8113/47/10/105301},
year = {2014},
volume = {47},
number = {10},
pages = {105301},
author = {Berry, Dominic W},
title = {High-order quantum algorithm for solving linear differential equations},
journal = {J. Phys. A: Math. Theor.},
}

@Inbook{Ashtekar_1999,
author="Ashtekar, Abhay
and Schilling, Troy A.",
editor="Harvey, Alex",
title="Geometrical Formulation of Quantum Mechanics",
bookTitle="On Einstein's Path: Essays in Honor of Engelbert Schucking",
year="1999",
publisher="Springer New York",
address="New York, NY",
pages="23--65",
isbn="978-1-4612-1422-9",
doi="10.1007/978-1-4612-1422-9_3",
url="https://doi.org/10.1007/978-1-4612-1422-9_3"
}

@article{Heslot_1985,
  title = {Quantum mechanics as a classical theory},
  author = {Heslot, Andr\'e},
  journal = {Phys. Rev. D},
  volume = {31},
  issue = {6},
  pages = {1341--1348},
  numpages = {0},
  year = {1985},
  doi = {10.1103/PhysRevD.31.1341},
}

@misc{Volovich_2025,
      title={Real Quantum Mechanics in a Kahler Space}, 
      author={Igor Volovich},
      year={2025},
      eprint={2504.16838},
      archivePrefix={arXiv},
      primaryClass={quant-ph},
      url={https://arxiv.org/abs/2504.16838}, 
}

@article{Brody_2001,
title = {Geometric quantum mechanics},
journal = {J. Geom. Phys.},
volume = {38},
number = {1},
pages = {19-53},
year = {2001},
issn = {0393-0440},
doi = {https://doi.org/10.1016/S0393-0440(00)00052-8},
author = {Dorje C. Brody and Lane P. Hughston},
}

@article{Kominis_2008,
doi = {10.1088/1751-8113/41/11/115202},
year = {2008},
month = {mar},
volume = {41},
number = {11},
pages = {115202},
author = {Kominis, Y and Hizanidis, K and Constantinescu, D and Dumbrajs, O},
title = {Explicit near-symplectic mappings of Hamiltonian systems with Lie-generating functions},
journal = {J. Phys. A: Math. Theor.}
}

@article{Strocchi,
  title = {Complex Coordinates and Quantum Mechanics},
  author = {Strocchi, F.},
  journal = {Rev. Mod. Phys.},
  volume = {38},
  issue = {1},
  pages = {36--40},
  numpages = {0},
  year = {1966},
  doi = {10.1103/RevModPhys.38.36}
}

@article{Briggs_2013,
  title = {Quantum dynamics simulation with classical oscillators},
  author = {Briggs, John S. and Eisfeld, Alexander},
  journal = {Phys. Rev. A},
  volume = {88},
  issue = {6},
  pages = {062104},
  numpages = {12},
  year = {2013},
  doi = {10.1103/PhysRevA.88.062104},
}

@article{Briggs_2012,
  title = {Coherent quantum states from classical oscillator amplitudes},
  author = {Briggs, John S. and Eisfeld, Alexander},
  journal = {Phys. Rev. A},
  volume = {85},
  issue = {5},
  pages = {052111},
  numpages = {10},
  year = {2012},
  doi = {10.1103/PhysRevA.85.052111},
}

@article{Mendes_2003,
  title = {Quantum control and the Strocchi map},
  author = {Vilela Mendes, R. and Man'ko, V. I.},
  journal = {Phys. Rev. A},
  volume = {67},
  issue = {5},
  pages = {053404},
  numpages = {8},
  year = {2003},
  doi = {10.1103/PhysRevA.67.053404}
}

@book{Feng_2010,
    author ={Kang Feng and Mengzhao Qin} ,
    title = {Symplectic Geometric Algorithms for Hamiltonian Systems},
    publisher ={Springer} ,
    year = {2010},
    doi={10.1007/978-3-642-01777-3}
}

@article{Volovich_2019,
    author ={Volovich, Igor V.} ,
    title = {Complete Integrability of Quantum and Classical Dynamical Systems},
    journal ={p-Adic Numbers Ultrametric Anal. Appl.} ,
    year = {2019},
    doi={10.1134/S2070046619040071},
    volume={11},
    issue={4},
    pages={328--334}
}

@article{Joseph_2020,
  title = {Koopman--von Neumann approach to quantum simulation of nonlinear classical dynamics},
  author = {Joseph, Ilon},
  journal = {Phys. Rev. Res.},
  volume = {2},
  issue = {4},
  pages = {043102},
  numpages = {17},
  year = {2020},
  doi = {10.1103/PhysRevResearch.2.043102}
}

@article{Bullock_2004,
author = {Bullock, Stephen S. and Markov, Igor L.},
title = {Asymptotically optimal circuits for arbitrary n-qubit diagonal comutations},
year = {2004},
volume = {4},
number = {1},
journal = {Quantum Inf. Comput.},
pages = {27–47},
numpages = {21},
doi={10.5555/2011572.2011575}
}

@article{Brizzard_2013,
title = {Jacobi zeta function and action-angle coordinates for the pendulum},
journal = {Commun. Nonlinear Sci. Numer. Simul.},
volume = {18},
number = {3},
pages = {511-518},
year = {2013},
doi = {https://doi.org/10.1016/j.cnsns.2012.08.023},
author = {Alain J. Brizard},
}

@article{Kaufman_1972,
    author = {Kaufman, Allan N.},
    title = {Quasilinear Diffusion of an Axisymmetric Toroidal Plasma},
    journal = {Phys. Fluids},
    volume = {15},
    number = {6},
    pages = {1063-1069},
    year = {1972},
    doi = {10.1063/1.1694031}
}

@book{Lichtenberg_1992,
    author = {A. J. Lichtenberg and M. A. Lieberman},
    title ={Regular and Chaotic Dynamics},
    publisher ={Springer},
    year = {1992},
    doi={10.1007/978-1-4757-2184-3}
}

@article{Kominis_2012,
  title = {Interaction of charged particles with localized electrostatic waves in a magnetized plasma},
  author = {Kominis, Y. and Ram, A. K. and Hizanidis, K.},
  journal = {Phys. Rev. E},
  volume = {85},
  issue = {1},
  pages = {016404},
  numpages = {8},
  year = {2012},
  doi = {10.1103/PhysRevE.85.016404}
}

@article{Kominis_2010,
  title = {Kinetic Theory for Distribution Functions of Wave-Particle Interactions in Plasmas},
  author = {Kominis, Y. and Ram, A. K. and Hizanidis, K.},
  journal = {Phys. Rev. Lett.},
  volume = {104},
  issue = {23},
  pages = {235001},
  numpages = {4},
  year = {2010},
  doi = {10.1103/PhysRevLett.104.235001}
}

@article{Cary_1981,
title = {Lie transform perturbation theory for Hamiltonian systems},
journal = {Phys. Reports},
volume = {79},
number = {2},
pages = {129-159},
year = {1981},
issn = {0370-1573},
doi = {https://doi.org/10.1016/0370-1573(81)90175-7},
author = {John R. Cary}
}

@misc{Wu_2025,
      title={Structure-preserving quantum algorithms for linear and nonlinear Hamiltonian systems}, 
      author={Hsuan-Cheng Wu and Xiantao Li},
      year={2025},
      eprint={2411.03599},
      archivePrefix={arXiv},
      primaryClass={quant-ph},
      url={https://arxiv.org/abs/2411.03599}, 
}

@book{Arnold_1989,
    author ={V. I. Arnold} ,
    title ={Mathematical Methods of Classical Mechanics} ,
    publisher ={Springer} ,
    year = {1989},
    doi={10.1007/978-1-4757-2063-1}
}

@article{Klimontovich_1957,
    author ={Y. Klimontovich} ,
    title ={ON THE METHOD OF "SECOND QUANTIZATION" IN PHASE SPACE} ,
    journal ={J. Exptl. Theoret. Phys. } ,
    year = {1957},
    volume={33},
    pages={982-990}
}

@article{Klimontovich_1972,
    author ={Y. Klimontovich} ,
    title ={Kinetic equations for classical nonideal plas-
mas} ,
    journal ={J. Exptl. Theoret. Phys. } ,
    year = {1972},
    volume={62},
    pages={1770-1781}
}

@article{Dougherty_1969,
title={The Statistical Theory of Non-Equilibrium Processes in a Plasma. By Yu. L. Klimontovich (translated by H. S. H. Massey and O. M. Blunn), Pergamon Press, 1967. 284 pp. 70s.}, 
volume={3}, DOI={10.1017/S0022377800004244}, 
number={1},
journal={J. Plasma Phys.}, 
author={Dougherty, J. P.}, 
year={1969},
pages={148–148}
}

@misc{Dsouza_2006,
      title={Quantum Simulation of Coupled Harmonic Oscillators: From Theory to Implementation}, 
      author={Viraj Dsouza and Weronika Golletz and Dimitrios Kranas and Bakhao Dioum and Vardaan Sahgal and Eden Schirman},
      year={2026},
      eprint={2603.05479},
      archivePrefix={arXiv},
      primaryClass={quant-ph},
      url={https://arxiv.org/abs/2603.05479}, 
}

@incollection{Brassard_2002,
   title={Quantum amplitude amplification and estimation},
   ISBN={9780821878958},
   ISSN={0271-4132},
   url={http://dx.doi.org/10.1090/conm/305/05215},
   DOI={10.1090/conm/305/05215},
   booktitle={Quantum Computation and Information},
   publisher={American Mathematical Society},
   author={Brassard, Gilles and Høyer, Peter and Mosca, Michele and Tapp, Alain},
   year={2002},
   editor={Samuel J. Lomonaco, Jr.},
   pages={53–74} 
   }

@article{Grikno_2021,
    author ={Dmitry Grinko and Julien Gacon and Christa Zoufal and Stefan Woerner} ,
    title ={Iterative quantum amplitude estimation} ,
    journal ={npj Quantum Inf.} ,
    year ={2021},
    volume={07},
    issue={52},
    doi={10.1038/s41534-021-00379-1}
}

@article{He_2015,
title = {Volume-preserving algorithms for charged particle dynamics},
journal = {J. Comput. Phys.},
volume = {281},
pages = {135-147},
year = {2015},
issn = {0021-9991},
doi = {10.1016/j.jcp.2014.10.032},
author = {Yang He and Yajuan Sun and Jian Liu and Hong Qin},
}

@article{Abdullaev_2002,
doi = {10.1088/0305-4470/35/12/307},
url = {https://doi.org/10.1088/0305-4470/35/12/307},
year = {2002},
volume = {35},
number = {12},
pages = {2811},
author = {S. S. Abdullaev},
title = {The 
Hamilton-Jacobi method and Hamiltonian maps},
journal = {J. Phys. A: Math. Theor.}
}

@book{Saletan_1998, 
place={Cambridge}, 
title={Classical Dynamics: A Contemporary Approach}, 
publisher={Cambridge University Press}, 
author={José, Jorge V. and Saletan, Eugene J.}, 
year={1998},
doi={10.1017/CBO9780511803772}
}

@misc{Manjarres,
      title={Hamiltonian point of view of quantum perturbation theory}, 
      author={A. D. Bermúdez Manjarres},
      year={2025},
      eprint={2107.07050},
      archivePrefix={arXiv},
      primaryClass={quant-ph},
      doi={10.48550/arXiv.2107.07050}, 
}

@book{Hairer_2006, 
place={Berlin, Heidelberg}, 
title={Geometric Numerical Integration: Structure-Preserving Algorithms for Ordinary Differential Equations}, 
publisher={Springer}, 
author={Ernst Hairer and Gerhard Wanner and Christian Lubich}, 
year={2006},
doi={10.1007/3-540-30666-8}
}

\end{document}